\pdfoutput=1

\documentclass[11pt]{article}

\usepackage[preprint]{acl/acl}

\usepackage{mathtools}
\usepackage{enumerate}
\usepackage{wrapfig}

\usepackage{caption}
\usepackage{amsthm}
\usepackage{multirow}

\newtheorem{theorem}{Theorem}

\usepackage{subcaption}
\usepackage{graphicx}

\usepackage{amssymb}
\usepackage[mathscr]{euscript}

\usepackage{natbib}

\usepackage{amsfonts,amsmath,amssymb}
\usepackage{bbm}
\usepackage{xspace}

\usepackage{multicol}
\usepackage{enumitem}

\usepackage{times}
\usepackage{latexsym}
\usepackage{booktabs}   
\usepackage{graphicx}   
\usepackage[ruled,vlined]{algorithm2e}

\usepackage[T1]{fontenc}

\usepackage[utf8]{inputenc}

\usepackage{microtype}

\usepackage{inconsolata}

\usepackage{graphicx}

\title{When prompt perturbations break your A/B test: \\ A valid statistical test for generative surveying}

\author{Hayden Helm \\
  Helivan \\ \texttt{hayden@helivan.io}
  \And Carey E. Priebe \\
  Johns Hopkins University}

\begin{document}
\maketitle

\begin{abstract}
Generative surveying -- where collections of LLM-based personas provide feedback on messages -- has emerged as a cheap and scalable alternative to traditional market research. 
However, LLMs are sensitive to small variations in prompt design and conclusions drawn from generative surveys may depend on arbitrary phrasing choices. 
Controlling for this sensitivity requires including semantically equivalent perturbations in the analysis.
In this paper, we show that standard hypothesis tests, including the sign test and Wilcoxon signed-rank test, are invalid under a statistical model for generative surveying that includes realistic perturbation structure. 
We propose a permutation test that is valid under this model and formally characterize the conditions under which standard tests fail.
Applying our framework to a simple generative surveying problem, we estimate relevant parameters, characterize the power of the permutation test under realistic conditions, and provide practical guidance on budget allocation across personas, perturbations, and replicates. 
Finally, we show that both the magnitude and direction of the estimated effect are sensitive to the choice of model, even within the same model family.
\end{abstract}

\section{Introduction}
\label{sec:introduction}

Generative surveying -- where a collection of LLM-based personas are queried to 
simulate population-level feedback on a product, message, or design -- has emerged 
as a cheap and scalable alternative to traditional market research 
\citep{argyle2023out, brand2023using, horton2023large}. 
The approach is attractive for the same reasons that LLMs are attractive as proxies 
for human subjects more broadly \citep{aher2023using, dillion2023can}: 
they are fast, cheap, and increasingly capable of producing human-like responses 
across a wide range of tasks \citep{helm2023statisticalturingtestgenerative,achiam2023gpt, dubey2024llama}. 
Given these advantages, generative surveying has been explored for preference 
elicitation \citep{brand2023using, hamalainen2023evaluating}, policy design simulation and prediction
\citep{park2024generative, helm2025digital}, and A/B testing of messages and 
designs \citep{brand2023using, yeykelis2024using}.

A well-known limitation of LLMs is their sensitivity to small variations in prompt 
wording \citep{lu2022fantastically, zhao2021calibrate, sclar2023quantifying, 
mizrahi2024state, ness2024medfuzz}. 
In the context of generative surveying, this sensitivity means that measured 
preference differences may depend on arbitrary phrasing choices rather than genuinely different preferences. 
Principled analysis of generative surveying must therefore control for this sensitivity by 
considering semantically equivalent perturbations. 

While this design is natural and well-motivated, it is also incompatible with the standard hypothesis tests used to 
analyze such data.
\begin{figure*}[t!]
    \centering
    \begin{subfigure}{0.49\linewidth}
        \centering
\includegraphics[width=\textwidth]{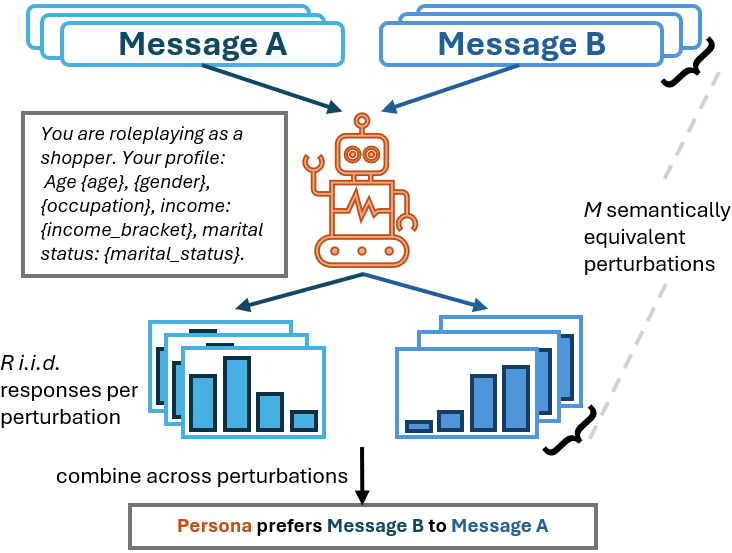}
    \end{subfigure}
    \begin{subfigure}{0.49\linewidth}
        \centering  {\includegraphics[width=\textwidth]{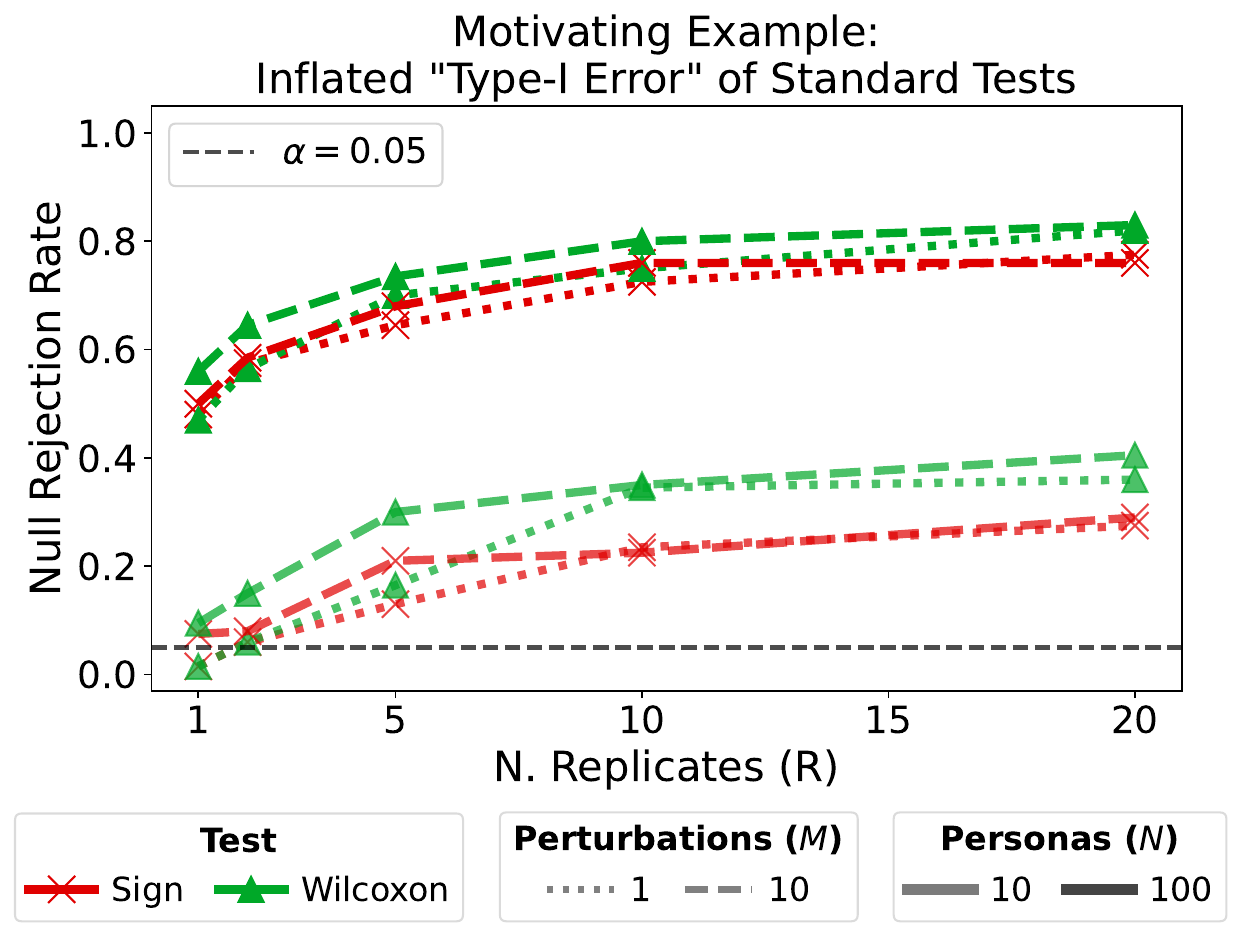}}
    \end{subfigure}
    \caption{(Left) Illustration of persona-based generative message testing for a single persona.
    The illustrated process is repeated across $ N $ personas and the preferences are aggregated for population-level analysis and statistical testing. 
    (Right) Standard paired tests have inflated Type-I error when persona-based message testing accounts for semantically equivalent message perturbations.}
    \label{fig:motivating-example}
\end{figure*}
Standard paired hypothesis tests (e.g., the sign test 
\citep{6e59f45d-d7d2-3290-89d3-655551349424, dixon1946statistical} and the Wilcoxon signed-rank test 
\citep{wilcoxon1945individual}) assume that cross-persona information is independent. 
Shared perturbations of the message, however, induce correlated preference shifts 
across personas which violates this assumption. 
As shown in Figure~\ref{fig:motivating-example} (right panel), the practical consequences are severe: null rejection rates exceed $0.8$ at a nominal level of $\alpha = 0.05$, with inflation growing with number of personas ($N$), number of perturbations ($M$), and and number of replicates per perturbation ($R$). 
Despite this, the effect of perturbations on standard paired tests in generative 
surveying has gone largely unexamined.
We make two primary contributions.
\begin{enumerate}[itemsep=0pt, topsep=0pt, parsep=2pt]
    \item \textbf{Statistical:} 
    We introduce a statistical model for binary generative surveys with an explicit cross-persona correlation parameter $\rho$ and propose a permutation test that is valid for all $\rho $. We formally characterize properties of the permutation test and the conditions under which standard paired tests fail. 
    \item \textbf{Empirical:} We estimate model parameters from a generative survey instantiated on ten models from the \texttt{Mistral-3} \citep{mistralai2025large3} and \texttt{Qwen-3} \citep{yang2025qwen3} families. 
    We show that $ \rho > 0 $ for the majority of models and that the proposed permutation test  does not exhibit the Type-I error inflation observed in standard tests.
    We further demonstrate that effect size and direction can vary across 
    models, even within the same family.
\end{enumerate}

\subsection{Motivating Example}
\label{sec:motivating-example}

We introduce an example that we return to throughout the paper.
Given that model use is often restricted to API-access only, we assume only black-box access throughout.

Consider a generative survey designed to determine whether a simulated population of consumer personas is more likely to purchase sneakers (Message $A$) or boots (Message $B$) in the black-box setting (see Figure \ref{fig:motivating-example}, left panel). 
We construct personas with demographic attributes (age, gender, 
occupation, income bracket, marital status; see Appendix \ref{app:personas} for examples) generated via \texttt{mistral-small-1225}. 
Each persona is defined via a system prompt that encodes a demographic profile; for example:

\small
\begin{quote}
\textit{You are roleplaying as a shopper. Your profile: Age 59, female, 
Professor, income: high, marital status: single. Answer ONLY `Yes' or `No'. 
Do not explain.}
\end{quote}

\normalsize
We use simple purchase-intent queries to mitigate potential confounds in message quality or length (e.g., ``I'd like to buy a pair of sneakers").
To control for LLM sensitivity to prompt wording, each message is represented 
by $M$ semantically equivalent perturbations -- paraphrases that vary 
vocabulary, formality, and sentence structure while preserving the conceptual 
and logical structure.
(e.g., ``I'm looking to purchase some new trainers''; see Appendix \ref{app:sneaker-perturbations} for full collection of perturbations for both messages).
In practice, the surveyor must determine what ``semantically equivalent perturbation" means.

Each persona is presented with each perturbation via 
the prompt

\small
\begin{quote}
\textit{A customer said: ``\{perturbation\}'' Based on this, would you personally 
be interested in buying the product they are describing?}
\end{quote}

\normalsize
\noindent and is queried with temperature $1.0$. 
In the black-box setting, the query-conditioned response distributions are not directly observable, so multiple replicates per (persona, perturbation) pair are necessary to estimate it.
We condition on responses that start with ``Yes" and ``No".
The process is repeated across all personas and their preferences are aggregated for population-level statistical testing.

To assess the appropriateness of standard paired tests in this setting, we construct a null condition by splitting 
$ 50 $ sneaker perturbations into two halves.
Since both sets of messages describe the same product, any rejection of the null hypothesis of no preference constitutes a Type-I error. 
Figure~\ref{fig:motivating-example} (right) shows the null rejection rates of 
the sign test \citep{dixon1946statistical} and Wilcoxon signed-rank test 
\citep{wilcoxon1945individual} across a range of $N$, $M$, and $R$ values. 
Both tests exhibit severely inflated Type-I error, with rejection rates exceeding 
$0.8$ at $\alpha = 0.05$.
The inflation grows with $M$ and $R$ and motivates the development of test that explicitly accounts for cross-persona correlation. 
We return to this example in \S\ref{sec:experiments}.

%

\section{Related Work}
\label{sec:related-work}
\begin{figure*}[t!]
    \centering
    \begin{subfigure}{0.4\linewidth}
        \centering
\includegraphics[width=\textwidth]{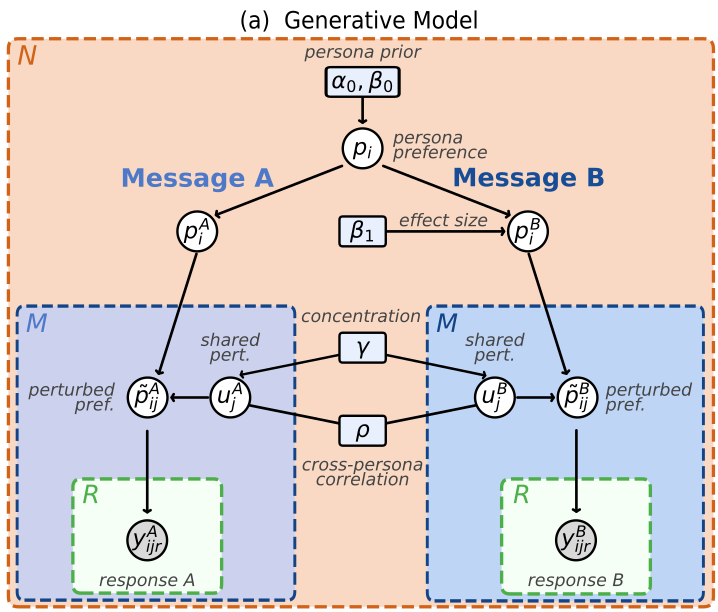}
    \end{subfigure}
    \begin{subfigure}{0.59\linewidth}
        \centering  {\includegraphics[width=\textwidth]{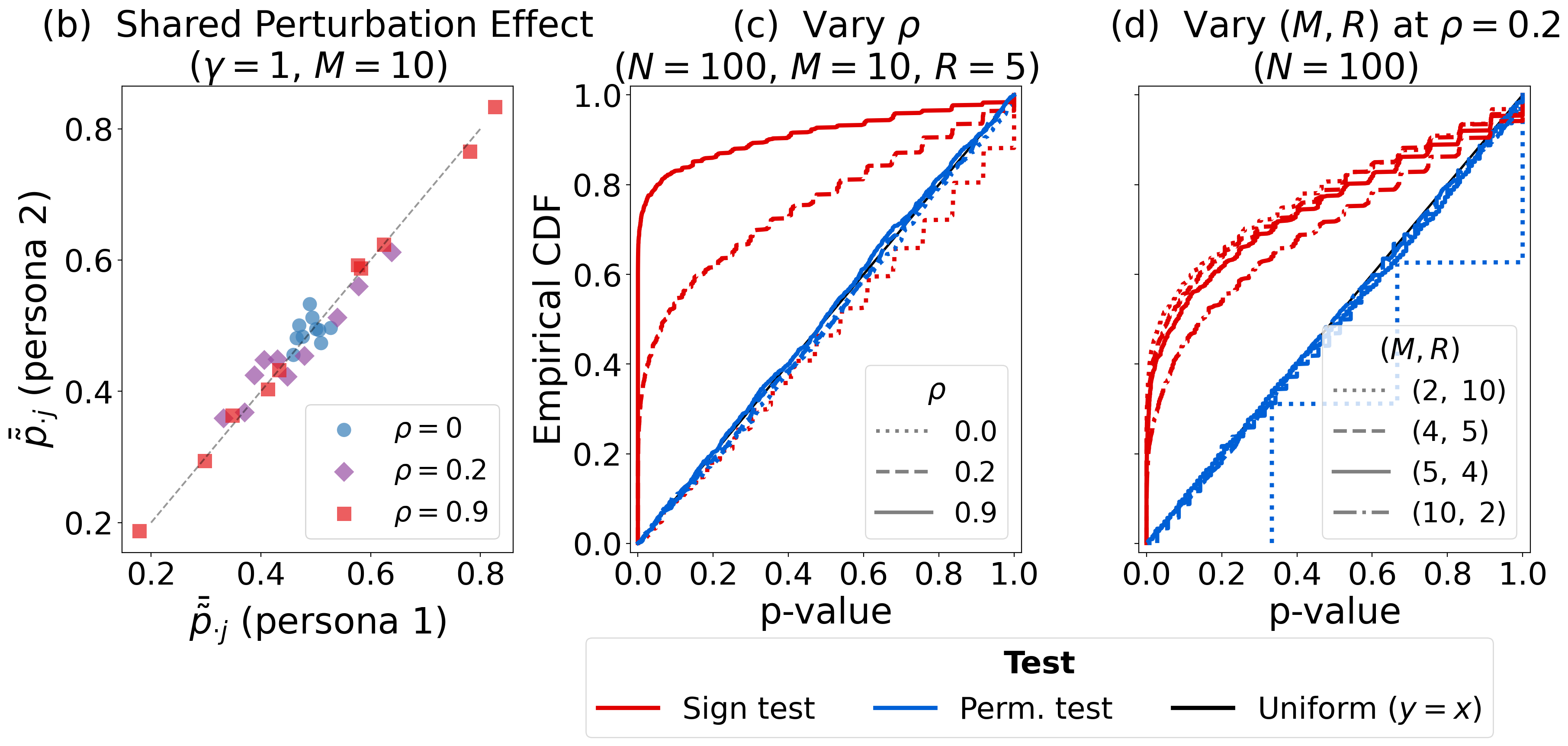}}
    \end{subfigure}
    \caption{(Left) Simple statistical model for persona-based generative surveying where outcomes are binary. 
    Light-grey boxes enclose model parameters. 
    Circles are generated parameters, dark-grey circles are observed data.
    (Center Left) Cross-persona perturbation correlation structure. Higher $ \rho $ means that a given perturbation is more likely to shift the preference  in the same direction across personas.
    (Center Right, Right) CDFs of the distribution of p-values for tests where semantically equivalent perturbations are treated as part of the ``null". 
    The sign test is invalid for $ \rho > 0 $. 
    The proposed permutation test is valid in all settings.
    }
    \label{fig:generative-model}
\end{figure*}

\paragraph{(Human) surveying sensitivity.} Surveying methodology has long 
recognized that response distributions are sensitive to factors unrelated to 
the construct being measured, including question wording 
\citep{sudman1974response, schuman1981questions, tourangeau2000psychology, 
krosnick2010question}, question ordering 
\citep{schuman1981questions, tourangeau2000psychology}, and interviewer 
effects \citep{groves1989survey, biemer2003introduction}. 
These sensitivities motivate careful experimental design, including split-sample 
designs to assess robustness to wording choices 
\citep{sudman1974response, schuman1981questions}. 
Our work is motivated by an analogous sensitivity in generative surveying, 
where LLM-based personas exhibit systematic variation across semantically 
equivalent prompt formulations.

\paragraph{Generative surveying.} Generative surveying has emerged as a popular 
approach for simulating population-level feedback using LLM-based personas 
\citep{argyle2023out, brand2023using, horton2023large}. 
Recent work has demonstrated that LLMs can approximate human survey responses 
across demographic groups and opinion domains \citep{ aher2023using, 
dillion2023can}, and has explored persona-based simulation for preference 
elicitation \citep{brand2023using, hamalainen2023evaluating}, policy evaluation 
\citep{park2024generative, helm2025digital}, and A/B testing 
\citep{brand2023using, yeykelis2024using}. 
The validity of LLMs as proxies for human subjects remains actively debated 
\citep{bisbee2024synthetic}, with concerns raised about demographic 
representativeness \citep{santurkar2023whose} and the stability of simulated 
opinions across prompt formulations \citep{sclar2023quantifying, mizrahi2024state}.
We propose a hypothesis test designed to be stable across prompt formulations and demonstrate its necessity for valid inference.

\paragraph{LLM prompt sensitivity.} The sensitivity of LLM outputs to prompt 
wording has been extensively documented 
\citep{zhao2021calibrate, lu2022fantastically, sclar2023quantifying, 
mizrahi2024state, ness2024medfuzz}, including in the context of evaluation 
benchmarks \citep{sclar2023quantifying, mizrahi2024state} and classification 
tasks \citep{zhao2021calibrate, lu2022fantastically}. 
Statistical approaches for handling this sensitivity include aggregating responses 
across prompt formulations \citep{mizrahi2024state} and testing for consistency 
across perturbations \citep{acharyya2025testing}. 

\section{Problem Setting \& Methodology}
\label{sec:method}

Let $f$ denote a generative model and let $\mathcal{Q}$ denote a set of 
queries. Given a query $q \in \mathcal{Q}$, the model produces a 
response $f(q) \in \mathcal{X} $ sampled from $ P_{f}(q) \in \mathcal{P} $.
Repeatedly querying the model $R$ times with 
the same query yields i.i.d.\ samples $f(q)_1, \ldots, f(q)_R$.
Finally, let $ p(P_{f}(q)) \in \mathbb{R} $ be a one-dimensional parameter of distribution $ P_{f} $ and $ \hat{p} $ be its plug-in estimate \citep{bickel1977mathematical}.

Different system prompts induce different response distributions for a given $ q $. 
Since a black-box model is completely characterized by its corresponding response distributions \citep{helm2026blackboxmodelclassificationdiscriminative}, different personas are different models.
A generative survey presents $ M $ different semantically equivalent perturbations of two queries, message $A$ and message $B$, represented at $ q^{A} $ and $ q^{B} $, respectively, to $ N $ personas.
Let $p^A_{ij} := p(P_{f_{i}}(q^{A}_{j}))$ 
and $p^B_{ij} := p(P_{f_{i}}(q^{B}_{j}))$ denote the parameter corresponding to persona $i$ and perturbation 
$j$ under each message. 

The primary goal is to determine whether the population-level preference for message $A$ differs from that of message $B$. 
Defining $\delta_i = \bar{p}^A_i - \bar{p}^B_i$ as the preference 
difference for persona $i$, where $\bar{p}^A_i = \frac{1}{M}\sum_{j=1}^M 
p^A_{ij}$ and $\bar{p}^B_i = \frac{1}{M}\sum_{j=1}^M p^B_{ij}$, the 
hypothesis of interest is
\begin{equation}
    H_0: \mathbb{E}[\delta_i] = 0 \quad \text{vs} \quad 
    H_A: \mathbb{E}[\delta_i] \neq 0.
    \label{eq:hypothesis}
\end{equation}
Given estimated persona preference differences $ \hat{\delta}_{1}, \hdots, \hat{\delta}_{N} $, the sign test and Wilcoxon signed-rank test are standard approaches for 
testing Eq.~\eqref{eq:hypothesis}.
Both tests assume that cross-persona observations are independent.

\subsection{A Statistical Model For Generative Surveying}
\label{subsec:generative-model}

We next describe a statistical model for generative surveying that makes the perturbation structure explicit.
The model is designed for the setting where $ p(P_{f}(q)) $ is the probability of success of a Bernoulli random variable and has three levels corresponding to personas, perturbations, and replicates. 

\paragraph{Persona level.} Each persona $i \in [N]$ has a latent preference 
$p_i \sim \text{Beta}(\alpha_0, \beta_0)$, representing their baseline 
probability of responding ``Yes''.

\paragraph{Perturbation level.} For each perturbation $j \in [M]$, the 
preference of persona $i$ is perturbed on the logit scale. The two messages 
share the same perturbation structure but differ by a population-level effect 
$\beta_1$:
\begin{align*}
    & \text{logit}(p^A_{ij}) = \text{logit}(p_i) + u_j + \varepsilon_{ij}, 
    \\
    & \text{logit}(p^B_{ij}) = \text{logit}(p_i) + \beta_1 + u_j + \varepsilon_{ij},
\end{align*}
where $u_j \sim \mathcal{N}(0, \sqrt{\rho/\gamma})$ is a shared perturbation 
effect common to all personas, and 
$\varepsilon_{ij} \sim \mathcal{N}(0, \sqrt{(1-\rho)/\gamma})$ is an independent 
effect specific to persona $i$ and perturbation $j$. The parameter 
$\gamma$ controls the total perturbation concentration, and 
$\rho \in [0, 1]$ is the amount of perturbation variance shared across personas 
(Figure~\ref{fig:generative-model}b). Under $H_0$, $\beta_1 = 0$ and the 
two messages are exchangeable; under $H_A$, $\beta_1 \neq 0$. When $\rho = 0$, 
perturbation effects are fully independent across personas. When $\rho > 0$, 
perturbation shifts preferences across personas in the same direction. It is 
precisely this shared shift that renders the sign test and Wilcoxon signed-rank 
test invalid.

\paragraph{Replicate level.} Each (persona, perturbation) pair is queried 
$R$ times independently:
\begin{equation*}
    Y_{ijr} \mid p_{ij} \sim \text{Bernoulli}(p_{ij}), 
    \quad r \in [R].
\end{equation*}
The generative structure is shown in Figure~\ref{fig:generative-model}a.
Example perturbed preference probabilities under different $ \rho $ are shown in Figure~\ref{fig:generative-model}b.
Analogous statistical models for other distributions with parameters in $ \mathbb{R} $ are possible.
The permutation test described next is directly applicable to these settings.

\subsection{Permutation Test}
\label{subsec:permutation-test}
We propose a permutation test that respects the cross-persona dependence structure induced by $\rho > 0$. 
For each perturbation $j \in [M]$, we compute the average estimated preference difference across personas:
\begin{equation*}
    \hat{d}_j = \frac{1}{N} \sum_{i=1}^N \bar{Y}^A_{ij\cdot} - 
    \bar{Y}^B_{ij\cdot},
\end{equation*}
where $\bar{Y}^A_{ij\cdot} = \frac{1}{R}\sum_{r=1}^R Y^A_{ijr}$. The test 
statistic is the mean of the perturbation-level differences:
\begin{equation*}
    T = \frac{1}{M} \sum_{j=1}^M \hat{d}_j.
\end{equation*}
Under $H_0$, the joint distribution of 
$(\hat{d}_1, \ldots, \hat{d}_M)$ is invariant to sign-flipping 
of the message labels. We approximate the null distribution of $T$ by 
randomly permuting the message labels across $\mathcal{B}$ independent sign-flip 
permutations and recomputing $T$ for each. The p-value is the proportion 
of permuted statistics at least as extreme as the observed $T$.
The test procedure for testing (i.e., Eq. \ref{eq:hypothesis}) is described in Algorithm \ref{alg:permutation-test}.

\begin{algorithm}[t]
\caption{Permutation Test for Generative Surveying}
\label{alg:permutation-test}
\KwIn{Responses $\{Y^A_{ijr}\}$ and $\{Y^B_{ijr}\}$ for $i \in [N]$, 
      $j \in [M]$, $r \in [R]$; significance level $\alpha$; 
      number of permutations $\mathcal{B}$}
\KwOut{p-value, reject/fail to reject $H_0$}

\tcp{Compute perturbation-level statistics}
\For{$j = 1, \ldots, M$}{
    $\bar{Y}^A_{ij\cdot} \leftarrow \frac{1}{R}\sum_{r=1}^R Y^A_{ijr}$ 
    for each $i \in [N]$\;
    $\bar{Y}^B_{ij\cdot} \leftarrow \frac{1}{R}\sum_{r=1}^R Y^B_{ijr}$ 
    for each $i \in [N]$\;
    $\hat{\delta}_j \leftarrow \frac{1}{N}\sum_{i=1}^N 
    \left(\bar{Y}^A_{ij\cdot} - \bar{Y}^B_{ij\cdot}\right)$\;
}

\tcp{Compute observed test statistic}
$T \leftarrow \frac{1}{M}\sum_{j=1}^M \hat{\delta}_j$\;

\tcp{Approximate null distribution via sign-flip permutations}
\For{$b = 1, \ldots, \mathcal{B}$}{
    Sample $\sigma^{(b)} = (\sigma^{(b)}_1, \ldots, \sigma^{(b)}_M)$ 
    with each $\sigma^{(b)}_j \overset{i.i.d.}{\sim} 
    \text{Uniform}\{-1, +1\}$\;
    $T^{(b)} \leftarrow \frac{1}{M}\sum_{j=1}^M 
    \sigma^{(b)}_j \hat{\delta}_j$\;
}

\tcp{Compute p-value}
$\text{p-value} \leftarrow \frac{1}{\mathcal{B}}\sum_{b=1}^{\mathcal{B}} 
\mathbf{1}\!\left[|T^{(b)}| \geq |T|\right]$\;

\KwRet{p-value; reject $H_0$ if p-value $\leq \alpha$}
\end{algorithm}

\paragraph{Validity.} Figure~\ref{fig:generative-model}c shows the empirical CDF of p-values under $H_0$ for the sign test and the permutation test across three values of $\rho$.
The CDF of a valid test tracks $ y = x $.
The CDF of an oversized test bows above it.
At $\rho = 0$, both tests are valid. 
For $\rho > 0$, the sign test CDF bows above the diagonal, with the deviation growing with $\rho$, confirming that $\rho > 0$ is sufficient for the sign test to be oversized. 
The permutation test tracks the diagonal across all values of $\rho$. 
Figure~\ref{fig:generative-model}d shows that the sensitivity of the sign test and the validity permutation test are robust to the allocation of budget across $M$ and $R$ at fixed $\rho$.

\subsection{Theoretical Properties}
\label{subsec:theory}

We next formalize that the sign test is invalid with cross-persona correlation ($ \rho > 0 $) and confirm the validity and consistency of the proposed permutation test.
We provide the full argument below each statement.

\begin{theorem}[Invalidity of the Sign Test]
\label{thm:invalidity}
Under $H_0$ and the generative model defined in 
Section~\ref{subsec:generative-model}. Assume $ \gamma > 0 $. Then the sign test has size strictly 
greater than $\alpha$ for all $ \alpha \in (0,1) $ and for all values of $\rho > 0$, $N \geq 2$, $M \geq 1$, and 
$R \geq 1$.
\end{theorem}
\begin{proof}[Proof]
Let $ \bar{Y}^{A}_{i\cdot\cdot} = \frac{1}{M}\sum_{j}^{M} \bar{Y}^{A}_{ij\cdot} $.
Let $ D_{i} = \hat{\bar{p}}^{A}_{i} - \hat{\bar{p}}_{i}^{B}$ and $ S = \sum_{i}^{N} \mathbf{1}[D_{i} > 0] $. 
The sign test treats the estimated per-persona preference difference indicators ($ \mathbf{1}[D_{i} > 0] $ as i.i.d.\ 
Bernoulli$(\frac{1}{2})$ under $H_0$, implying $\text{Var}_{\text{sign}}(S) 
= N/4$. When $\rho > 0$ and $ \gamma > 0 $, the shared perturbation effect $u_j$ induces 
strictly positive covariance between $p_{ij}$ and $p_{i'j}$ 
for $i \neq i'$, which propagates to the aggregated preferences and 
consequently to the preference indicators. The true variance therefore 
satisfies $\text{Var}_{\text{true}}(S) > N/4$, so the sign test 
systematically underestimates the variance of its statistic, uses critical 
values that are too small, and rejects $H_0$ too often. 
\end{proof}

\begin{theorem}[Validity of the Permutation Test]
\label{thm:validity}
Under $H_0$ and the generative model defined in 
Section~\ref{subsec:generative-model}, the permutation test has exact size 
$\alpha$ for all $\alpha \in (0, 1)$ and all values of $N \geq 1$, $M \geq 1$, 
$R \geq 1$, $\rho \in [0, 1]$, and $\gamma > 0$.
\end{theorem}

\begin{proof}[Proof]
Under $H_0$, $\beta_1 = 0$ implies $p^A_{ij} = p^B_{ij}$ for all $i, j$, 
so swapping the $A$ and $B$ labels for perturbation $j$ negates 
$\hat{d}_j$ without changing its distribution. Since this holds jointly 
across all $j \in [M]$, the vector $(\hat{d}_1, \ldots, \hat{d}_M)$ 
is exchangeable under sign-flipping of the message labels, and the 
permutation distribution of $T$ is the exact null distribution. Therefore 
$\Pr(\text{p-value} \leq \alpha \mid H_0) = \alpha$ for all parameter values. 
\end{proof}

\begin{theorem}[Consistency of the Permutation Test]
\label{thm:consistency}
Under $H_A: \mathbb{E}[\delta_i] \neq 0$ and the generative model defined 
in Section~\ref{subsec:generative-model}, the power of the permutation test 
converges to $1$ as $N \to \infty$ for any fixed $M \geq 1$, $R \geq 1$, 
$\rho \in [0, 1]$, $\gamma > 0$, and $\beta_1 \neq 0$.
\end{theorem}

\begin{proof}[Proof]
The permutation critical value converges in probability to $|\beta_1| \cdot q_{1-\alpha}$, 
where $q_{1-\alpha} < 1$ is the $(1-\alpha)$ quantile of a Rademacher average. Since 
$T \xrightarrow{p} \beta_1$ while the critical value converges to a strictly smaller 
quantity, the rejection probability converges to $1$.
\end{proof}

\section{Experiments}
\label{sec:experiments}

We apply the framework from \S\ref{sec:method} to the sneakers versus 
boots example introduced in \S\ref{sec:motivating-example}. 
\begin{figure}
    \centering
    \includegraphics[width=0.95\linewidth]{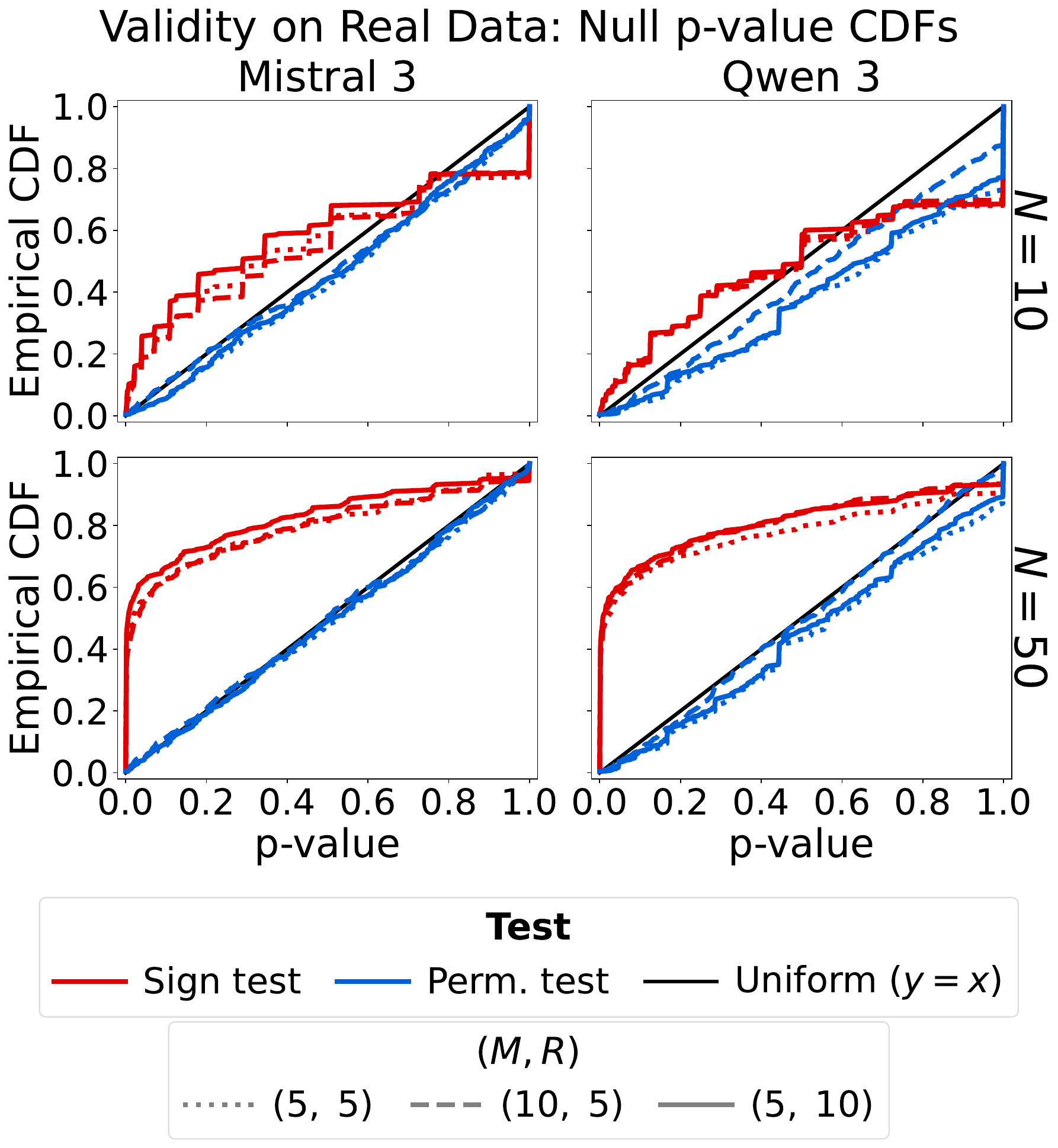}
    \caption{Empirical CDF of p-values under the null condition (sneakers vs.\ sneakers 
split) for the sign test and permutation test on real data. Each curve shows 
the median CDF across models in each family. The sign test is oversized across 
all families, configurations, and values of $N$. The permutation test is valid 
at $N = 50$ and mildly conservative at $N = 10$ due to the coarseness of the 
permutation distribution at small $M$.}
    \label{fig:validity-real}
\end{figure}

\subsection{Experimental Setup}
\label{subsec:setup}

We collect data across ten models spanning two families: \texttt{Mistral-3} and \texttt{Qwen-3}
(see \ref{app:models} for complete model list and generation parameters). 
For each experimental configuration $(N, M, R)$, we generate 200 random samples and compute rejection rates at $\alpha = 0.05$.
We provide estimates of Type-I error under the null condition (sneakers vs.\ sneakers split) and power under the alternative condition (sneakers vs.\ boots).

\subsection{Validity on Real Data}
\label{subsec:validity}

To verify that the permutation test controls Type-I error on real data, we 
compute the empirical CDF of p-values under the null condition for both the 
sign test and the permutation test. For each model and configuration $(N, M, 
R)$, we draw 200 sub-samples and record the p-value of each test. We report the median empirical CDF across models within each family. 
Individual model rejection rates across all configurations are reported in 
Table~\ref{tab:validity} in Appendix \ref{app:validity}.

Figure~\ref{fig:validity-real} shows results for three allocations -- $(M, R) 
\in \{(5, 5), (10, 5), (5, 10)\}$ -- at $N \in \{10, 50\}$ for both families. 
The sign test is oversized across all families, configurations, and values of 
$N$, with deviation from the diagonal growing with $M$, consistent with the 
simulation results in Figure~\ref{fig:generative-model}. The permutation test 
tracks the diagonal closely at $N = 50$. At $N = 10$, the permutation test is 
mildly conservative -- the empirical CDF falls below the diagonal at small 
p-values -- a consequence of the coarseness of the permutation distribution 
when $M$ is small: with $M = 5$ perturbations, the permutation distribution 
has limited resolution, producing a discrete set of attainable p-values that 
cannot approximate the uniform distribution arbitrarily closely. This 
conservatism is a property of the test at small $M$ rather than a failure of 
validity, and disappears as $M$ grows.

\subsection{Power Profiling}
\label{subsec:power}

We characterize the power of the permutation test on real data by computing 
rejection rates under the alternative condition (sneakers vs.\ boots) via 
the same sub-sampling procedure used for validity assessment. Power varies 
substantially across models, driven primarily by differences in effect size. 
\texttt{Mistral-3-14B} and \texttt{Qwen-3-14B} achieve the highest power, reaching $1.00$ at 
moderate budgets under the $N=50$, $M=20$, $R=5$ allocation. Models with 
near-degenerate response distributions (i.e., where $\Pr[\text{Yes}] \approx 0$ or $\Pr[\text{Yes}] \approx 1$; seen in \texttt{Qwen-3-0.6B} through \texttt{-4B} and \texttt{Mistral-3-
Large}) exhibit low power regardless of budget allocation, reflecting small 
or near-zero effect sizes rather than limitations of the test. Across all 
models, allocating budget toward perturbations ($M$) rather than replicates 
($R$) consistently yields higher power: comparing $M=20, R=5$ against $M=5, 
R=20$ at fixed $N$ shows substantial gains for nearly every model. Power 
estimates are reported in Table~\ref{tab:parameters_power}.

\subsection{Parameter Estimation}
\label{subsec:parameter-estimation}
\begin{table*}[t]
\centering
\caption{Estimated generative model parameters and permutation test power at varying budgets
$B = N \times M \times R$ for each model. Standard errors of the estimates are provided in parentheses.
Power columns denote $M$/$R$ allocations for each $N$.
\texttt{Qwen3-0.6B} is degenerate (always responds ``yes'') and parameters cannot be estimated.}
\label{tab:parameters_power}
\setlength{\tabcolsep}{4.75pt}
\footnotesize
\begin{tabular}{lrrrr | rrr | rrr}
\toprule
& & & & &
  \multicolumn{3}{c}{$N = 10$} &
  \multicolumn{3}{c}{$N = 50$} \\
\cmidrule(lr){6-8} \cmidrule(lr){9-11}
Model
  & $\hat\alpha_0/(\hat\alpha_0\!+\!\hat\beta_0)$ & $\hat\alpha_0\!+\!\hat\beta_0$ & $\hat{\gamma}$ & $\hat{\rho}$
  & \textit{5/5} & \textit{5/20} & \textit{20/5}
  & \textit{5/5} & \textit{5/20} & \textit{20/5} \\
\midrule
\textbf{Mistral} & & & & & & & & & & \\
\quad \textit{3B} & $0.45 (0.03)$ & $1.79 (0.15)$ & $1.01 (0.16)$ & $0.48 (0.06)$ & 0.11 & 0.14 & 0.42 & 0.10 & 0.11 & 0.35 \\
\quad \textit{8B} & $0.67 (0.04)$ & $1.60 (0.23)$ & $0.36 (0.05)$ & $0.46 (0.04)$ & 0.16 & 0.18 & 0.49 & 0.16 & 0.16 & 0.67 \\
\quad \textit{14B} & $0.38 (0.03)$ & $1.98 (0.25)$ & $0.40 (0.05)$ & $0.45 (0.05)$ & 0.47 & 0.52 & 0.85 & 0.49 & 0.49 & 1.00 \\
\quad \textit{Large (41B act.)} & $0.07 (0.02)$ & $11.12 (33.75)$ & $1.06 (0.37)$ & $0.16 (0.05)$ & 0.15 & 0.18 & 0.49 & 0.17 & 0.17 & 0.57 \\
\midrule
\textbf{Qwen 3} & & & & & & & & & & \\
\quad \textit{0.6B} & --- & --- & --- & --- & 0.00 & 0.00 & 0.00 & 0.00 & 0.00 & 0.00 \\
\quad \textit{1.7B} & $0.97 (0.00)$ & $107.26 (106.15)$ & $1.34 (0.50)$ & $0.20 (0.04)$ & 0.01 & 0.01 & 0.02 & 0.03 & 0.03 & 0.03 \\
\quad \textit{4B} & $0.97 (0.01)$ & $103.47 (149.12)$ & $1.10 (0.47)$ & $0.09 (0.03)$ & 0.01 & 0.01 & 0.03 & 0.03 & 0.03 & 0.00 \\
\quad \textit{8B} & $0.75 (0.04)$ & $2.13 (0.51)$ & $0.20 (0.03)$ & $0.39 (0.06)$ & 0.04 & 0.04 & 0.20 & 0.04 & 0.03 & 0.07 \\
\quad \textit{14B} & $0.52 (0.04)$ & $1.16 (0.11)$ & $0.16 (0.02)$ & $0.41 (0.05)$ & 0.62 & 0.61 & 1.00 & 0.66 & 0.67 & 1.00 \\
\quad \textit{32B} & $0.91 (0.01)$ & $18.21 (6.34)$ & $0.49 (0.06)$ & $0.31 (0.05)$ & 0.27 & 0.28 & 0.73 & 0.28 & 0.29 & 0.90 \\
\bottomrule
\bottomrule
\end{tabular}
\end{table*}

We estimate the parameters of the generative model from 
Section~\ref{subsec:generative-model} separately for each model using the 
null condition data, with bootstrap standard errors computed by resampling 
both personas and perturbations with replacement ($B = 1000$). Parameters 
are reported in Table~\ref{tab:parameters_power} using the mean $\hat{\alpha}_0 / 
(\hat{\alpha}_0 + \hat{\beta}_0)$ and precision $\hat{\alpha}_0 + \hat{\beta}_0$ 
of the Beta prior. We highlight two observations. First, bootstrap confidence 
intervals for $\hat{\rho}$ exclude zero for all estimable models, confirming 
that $\rho > 0$ is the realistic regime and that the sign test will be 
oversized in practice. Second, models with near-degenerate response 
distributions exhibit large standard errors on the Beta prior precision, reflecting poor 
identifiability of the prior shape parameters in these regimes.

\subsection{Budget Allocation}
\label{subsec:budget}

To provide practical guidance on budget allocation, we simulate power curves 
as a function of total budget $ N \times M \times R$ under eight $N$:$M$:$R$ 
allocation strategies, across four parameter settings spanning the range 
observed in Table~\ref{tab:parameters_power}. Simulation parameters are set 
to the median values estimated from real data. As shown in 
Figure~\ref{fig:budget}, allocating budget toward perturbations consistently 
maximizes power: the $1$:$10$:$1$ strategy dominates across nearly all 
parameter settings. Increasing $\rho$ substantially reduces achievable power 
at fixed budget, while $\gamma$ has a comparatively modest effect. These 
results provide actionable guidance for practitioners designing generative 
surveys: given a fixed query budget, perturbations should be prioritized over 
replicates or additional personas.

\subsection{The Model Matters}
\label{subsec:model-matters}
\begin{figure*}[h]
    \centering
    \includegraphics[width=0.95\linewidth]{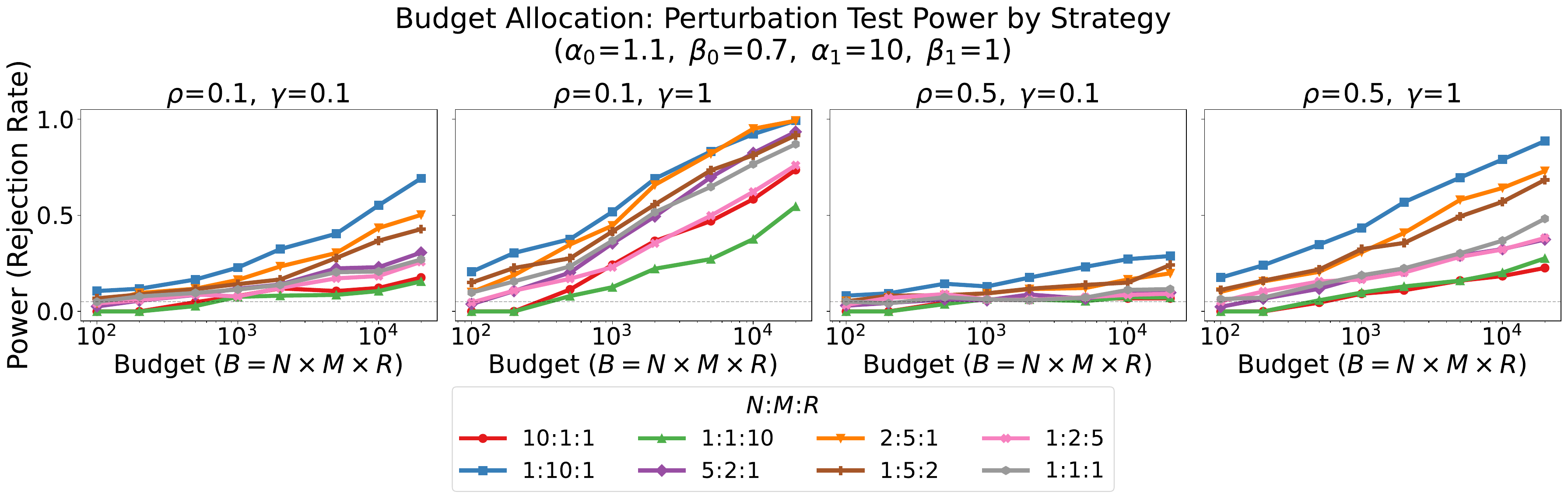}
    \caption{Power of the permutation test as a function of total query budget $N \times M \times R$ under eight budget allocation strategies (N:M:R), across four parameter settings ($\rho \in \{0.1, 0.5\}$, $\gamma \in \{0.1, 1\}$) spanning the range estimated from real data (Table~\ref{tab:parameters_power}). Simulations use median parameter values estimated from real data. Allocating budget toward perturbations ($M$) rather than personas ($N$) or replicates ($R$) consistently yields the highest power. Increasing $\rho$ substantially reduces achievable power at fixed budget, while $\gamma$ has a comparatively modest effect.}
    \label{fig:budget}
\end{figure*}

A natural question in generative surveying is whether the choice of model 
affects downstream conclusions. Figure~\ref{fig:model-matters} shows the 
estimated effect size as a function of active parameters for all ten models 
across the \texttt{Mistral-3} and \texttt{Qwen-3} families. 
We highlight two observations.
\begin{figure}[h]
    \centering
    \includegraphics[width=0.95\linewidth]{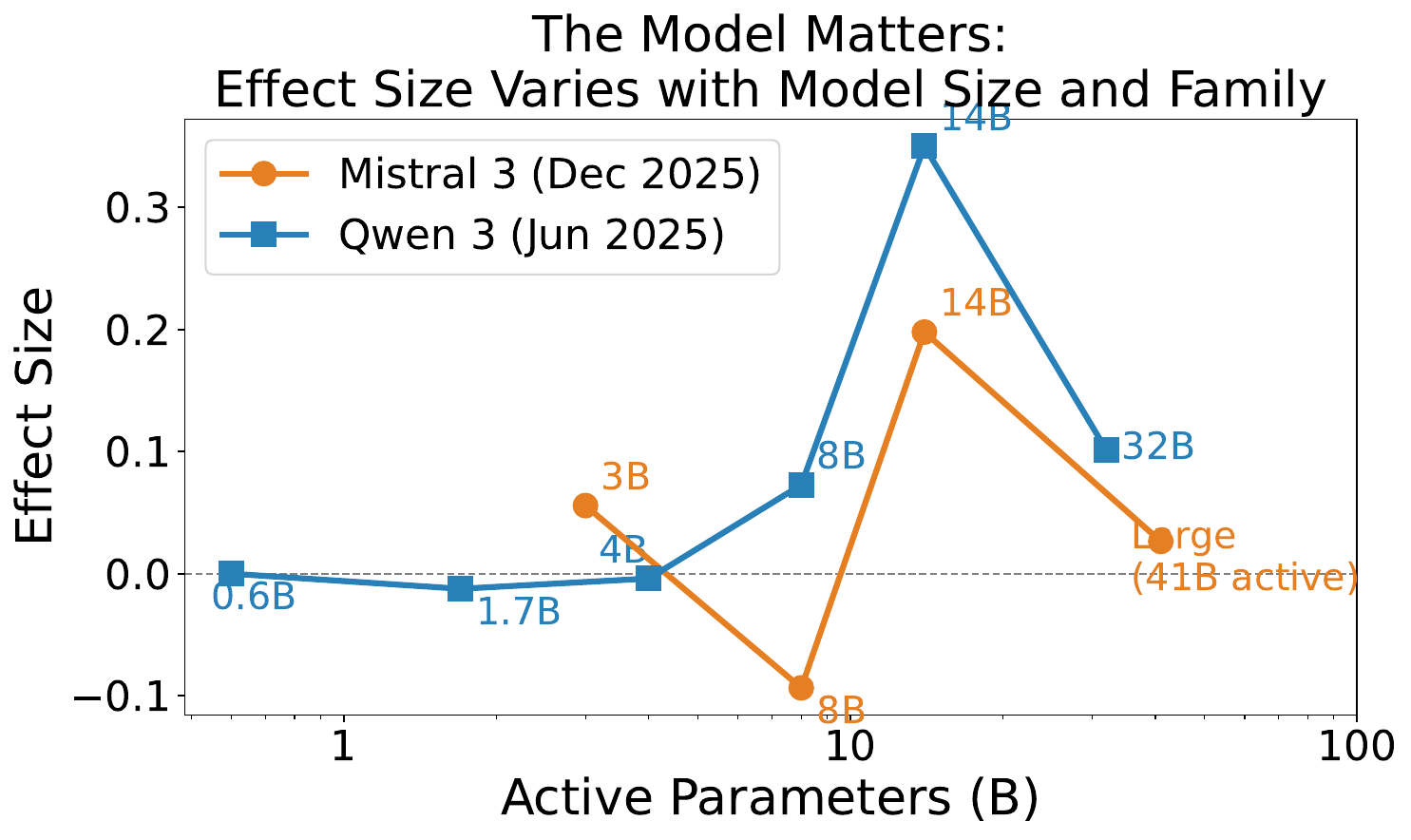}
    \caption{Effect size ($\hat{\beta}_1$) versus number of active parameters for 
models in the \texttt{Mistral-3} and \texttt{Qwen-3} families. Effect size is non-monotonic 
in model scale within both families. Same scale models across families can 
produce estimates of opposite sign (e.g., \texttt{Mistral-8B} and \texttt{Qwen-3-8B} estimate 
effects in opposite directions) meaning that model choice alone can determine the conclusion of a generative survey. 
}
    \label{fig:model-matters}
\end{figure}

First, effect size is non-monotonic in model scale within both families. 
Larger models do not systematically produce larger or more consistent effect 
estimates -- the relationship between scale and effect size is irregular, 
with notable peaks at 14B for both families followed by sharp declines. 

Second, same-scale models across families can produce 
effect estimates of opposite sign. \texttt{Mistral-8B} estimates a negative effect 
($\hat{\beta}_1 < 0$, preferring sneakers) while \texttt{Qwen-3-8B} estimates a 
positive effect ($\hat{\beta}_1 > 0$, preferring boots).  
Together, these observations suggest that results are highly conditioned on the model, not just the model family or size.

We note that our results are conditioned on a specific prompt structure and persona 
construction. 
With that said, given that the direction and magnitude of effect is highly conditioned on the model, we recommend that practitioners report the 
model used in any generative survey, treat conclusions as model-conditional, 
and replicate findings across multiple models before acting on 
results when possible.

\section{Discussion}
\label{sec:discussion}

We showed that standard paired hypothesis tests are invalid for generative 
surveying when perturbations induce cross-persona correlation ($\rho > 0$), 
introduced a simple statistical model that includes cross-persona correlation, proposed a permutation test that is valid for all $\rho \geq 0$, and 
validated the framework empirically across ten models in the \texttt{Mistral-3} and 
\texttt{Qwen-3} families. 
Our parameter estimates confirmed that $\rho > 0$ is the 
most typical regime in real data and showed that perturbations should be prioritized over replicates to maximize power under a fixed budget. 
We additionally found that effect size and direction are sensitive to model choice even within the same family, underscoring the need for careful reporting standards in generative surveying.

\subsection*{Limitations and Future Work}

\noindent\textbf{Binary responses.} We focus on binary responses for simplicity. 
The permutation test extends naturally to any response structure for which a perturbation-level summary statistic can be computed (e.g., mean ratings) and the generative model can be adapted accordingly. 
Extending the empirical and theoretical results to non-binary settings is an important direction for 
future work.

\noindent\textbf{Prompt structure and persona construction.} The estimated 
parameters and effect sizes reported in Section~\ref{sec:experiments} are 
conditioned on a specific prompt structure and persona construction. 
Different system prompts, persona attribute sets, or persona generation 
procedures may produce different values of $\rho$, $\gamma$, and $\beta_1$. 
As shown by the effect of model size within a model family, conclusions should not be generalized across substantially different designs without further validation or  
re-estimation.

\noindent\textbf{Joint versus separate A/B presentation.} Our framework assumes each message is presented to each persona independently.
A natural 
alternative is to 
present message $ A $ and message $ B $ within a single prompt. 
However, LLMs are sensitive to the ordering of options within a prompt 
\citep{ness2024medfuzz}, requiring both orderings to be presented to control for the presentation-order confound -- doubling the number of API calls with no 
reduction in total cost relative to independent presentation. 
The separate-presentation design additionally 
supports multi-way comparisons efficiently: responses to message $A$ 
collected under the $N \times M \times R$ design can be reused for any 
subsequent comparison -- whether A vs.\ B or A vs.\ C -- without additional API calls, so total cost scales with the 
number of messages rather than the number of pairs.

\noindent\textbf{Recommendations.} Based on our results, 
we offer three recommendations. First, always use $M > 1$ semantically 
equivalent perturbations and apply the permutation test rather than the 
sign test or Wilcoxon signed-rank test. Second, allocate budget toward 
perturbations ($M$) rather than replicates ($R$). Third, treat conclusions as model-conditional. 
\label{introduction}

\label{method}


\label{experiments}

\clearpage

\bibliography{biblio}

\clearpage

\appendix
\onecolumn
\appendix

\section{Models \& Perturbation and Persona Examples}
\label{app:examples}

\subsection{Models and Generation Parameters}
\label{app:models}

Mistral models are queried via the Mistral Batch API with pinned version IDs 
for reproducibility. Qwen 3 models are queried via the DashScope API with 
\texttt{enable\_thinking=false} to disable chain-of-thought reasoning. All 
models are queried with temperature $1.0$ and \texttt{max\_tokens=1}, 
yielding $N \times (M^A + M^B) \times R = 100 \times 75 \times 20 = 
150{,}000$ requests per model and $1{,}500{,}000$ requests in total.

\begin{table}[h]
\centering
\caption{Mistral 3 models queried via the Mistral Batch API.}
\label{tab:mistral-models}
\small
\begin{tabular}{llr}
\toprule
Model & API ID & Parameters \\
\midrule
Ministral 3B  & \texttt{ministral-3b-2512}    & 3B \\
Ministral 8B  & \texttt{ministral-8b-2512}    & 8B \\
Ministral 14B & \texttt{ministral-14b-2512}   & 14B \\
Mistral Large & \texttt{mistral-large-2512}   & 41B active / 675B total (MoE) \\
\bottomrule
\end{tabular}
\end{table}

\begin{table}[h]
\centering
\caption{Qwen 3 models queried via the DashScope API with 
\texttt{enable\_thinking=false}.}
\label{tab:qwen-models}
\small
\begin{tabular}{llr}
\toprule
Model & API ID & Parameters \\
\midrule
Qwen3 0.6B & \texttt{qwen3-0.6b}  & 0.6B \\
Qwen3 1.7B & \texttt{qwen3-1.7b}  & 1.7B \\
Qwen3 4B   & \texttt{qwen3-4b}    & 4B \\
Qwen3 8B   & \texttt{qwen3-8b}    & 8B \\
Qwen3 14B  & \texttt{qwen3-14b}   & 14B \\
Qwen3 32B  & \texttt{qwen3-32b}   & 32B \\
\bottomrule
\end{tabular}
\end{table}

\subsection{Persona Examples}
\label{app:personas}

We generate $N = 100$ synthetic shopper personas via 
\texttt{mistral-small-latest}. Each persona has five demographic attributes: 
age, gender, occupation, income bracket, and marital status. The first 100 
personas are used in all experiments. Table~\ref{tab:personas} shows a 
representative sample.

\begin{table}[h]
\centering
\caption{Representative sample of synthetic shopper personas.}
\label{tab:personas}
\small
\begin{tabular}{rrlllll}
\toprule
ID & Age & Gender & Occupation & Income & Marital Status \\
\midrule
0  & 25 & male       & Software Engineer    & high   & single   \\
1  & 32 & female     & Nurse                & medium & married  \\
2  & 45 & non-binary & Graphic Designer     & medium & divorced \\
3  & 60 & male       & Retired Teacher      & low    & widowed  \\
4  & 28 & female     & Marketing Manager    & high   & single   \\
5  & 55 & male       & Construction Worker  & medium & married  \\
6  & 37 & non-binary & Social Worker        & low    & single   \\
7  & 42 & female     & Lawyer               & high   & divorced \\
8  & 65 & male       & Retired Mechanic     & low    & widowed  \\
9  & 30 & female     & Chef                 & medium & single   \\
\bottomrule
\end{tabular}
\end{table}

\subsection{Sneaker Perturbations ($M = 50$)}
\label{app:sneaker-perturbations}

All 50 semantically equivalent paraphrases of the sneaker purchase-intent 
query, generated and validated via \texttt{mistral-small-latest}:

\small
\begin{enumerate}[itemsep=0pt, topsep=0pt, parsep=1pt]
\item I'd like to buy a pair of sneakers.
\item I'm looking to purchase some new trainers.
\item I want to shop for a pair of athletic casual shoes.
\item I'm considering buying sneakers.
\item I'm in the market for a new pair of sneakers.
\item I'm thinking about purchasing some trainers.
\item I'd like to get myself a pair of sneakers.
\item I'm interested in buying athletic casual shoes.
\item I want to buy some new sneakers.
\item I'm planning to purchase a pair of trainers.
\item I'm looking to get a pair of sneakers.
\item I'd like to invest in some new athletic casual shoes.
\item I'm considering buying a pair of sneakers.
\item I want to shop for trainers.
\item I'm interested in purchasing sneakers.
\item I'd like to grab a pair of athletic casual shoes.
\item I'm thinking about getting some new sneakers.
\item I'm in the mood to buy trainers.
\item I'm looking to acquire a pair of sneakers.
\item I want to pick up some new athletic casual shoes.
\item I'm considering shopping for sneakers.
\item I'd like to obtain a pair of trainers.
\item I'm planning to get myself some new sneakers.
\item I want to purchase a pair of trainers.
\item I'm looking to purchase some athletic shoes.
\item Can you show me some sneakers to buy?
\item I want to shop for a new pair of trainers.
\item I'm interested in buying some casual athletic shoes.
\item Do you have sneakers available for purchase?
\item I'm thinking about getting a pair of sneakers.
\item I'd like to explore options for buying sneakers.
\item I want to find a pair of sneakers to buy.
\item I'm in the market for some new trainers.
\item I'm considering purchasing a pair of sneakers.
\item I'd like to buy some comfortable athletic shoes.
\item I'm looking to add sneakers to my collection.
\item I want to check out sneakers to purchase.
\item I'm interested in buying a pair of running shoes.
\item I'd like to shop for some stylish sneakers.
\item I'm hoping to find sneakers to buy soon.
\item I want to purchase a pair of casual sneakers.
\item I'm considering buying some athletic footwear.
\item I'd like to explore sneaker options for purchase.
\item I want to buy some new trainers.
\item I'm interested in purchasing some sneakers.
\item I'd like to add a pair of sneakers to my wardrobe.
\item I'm considering buying some casual athletic shoes.
\item I'd like to buy a new pair of sneakers.
\item Could you show me some trainers that are available for purchase?
\item I'm looking to shop for a pair of athletic casual shoes.
\end{enumerate}

\normalsize
\subsection{Boot Perturbations ($M = 25$)}
\label{app:boot-perturbations}

All 25 semantically equivalent paraphrases of the boot purchase-intent 
query, generated and validated via \texttt{mistral-small-latest}:

\small
\begin{enumerate}[itemsep=0pt, topsep=0pt, parsep=1pt]
\item I'd like to buy a pair of boots.
\item I'm looking to purchase some boots.
\item I want to buy boots.
\item I'm in the market for a new pair of boots.
\item I'd like to shop for boots.
\item I'm interested in getting boots.
\item I need to buy boots.
\item I'm considering purchasing boots.
\item I'm planning to buy a pair of boots.
\item I'm thinking about buying boots.
\item I'm hoping to purchase boots soon.
\item I'm looking for boots to buy.
\item I'm ready to buy boots.
\item I'd like to find boots to purchase.
\item I'm shopping for boots.
\item I'm aiming to buy boots.
\item I'm set on purchasing boots.
\item I'm eager to buy a pair of boots.
\item I'm on the hunt for boots to buy.
\item I'm keen to purchase boots.
\item I'd like to buy some boots.
\item I'm looking for a pair of boots to purchase.
\item I'm planning to shop for boots soon.
\item I want to get myself a pair of boots.
\item I'm thinking about buying boots, specifically a pair.
\end{enumerate}

\normalsize

\section{Individual Model Validity}
\label{app:validity}

Table~\ref{tab:validity} reports the null rejection rate at $\alpha = 0.05$ 
for the sign test and permutation test for each individual model under the 
null condition (sneakers vs.\ sneakers split), across six $(N, M, R)$ 
configurations. Each cell is based on 200 random sub-samples. The sign test 
is oversized for every model with a non-degenerate response distribution. 
The permutation test controls size near $\alpha = 0.05$ across all models 
and configurations. Sign test inflation grows substantially with $N$: at 
$N = 50$, rejection rates reach $0.5$--$0.7$ for most models. Mistral 
Large at $N = 10$ is an exception -- the near-zero base rate 
($\hat{p}_i \approx 0.046$) produces near-zero variance in the response 
distribution, suppressing the sign test statistic regardless of the 
perturbation structure. Qwen 3 0.6B is degenerate (always responds 
``Yes'') and produces zero rejection rates for both tests.
\begin{table}[h]
\centering
\caption{Null rejection rates at $\alpha = 0.05$ for the sign test / permutation 
test under the null condition (sneakers vs.\ sneakers split), shown as 
Sign / Perm.\ for each $(N, M, R)$. Each cell is based on 200 random sub-samples.}
\label{tab:validity}
\small
\setlength{\tabcolsep}{4pt}
\begin{tabular}{l cccccc}
\toprule
& \multicolumn{3}{c}{$N = 10$} & \multicolumn{3}{c}{$N = 50$} \\
\cmidrule(lr){2-4} \cmidrule(lr){5-7}
Model & $(5,5)$ & $(5,10)$ & $(10,5)$ & $(5,5)$ & $(5,10)$ & $(10,5)$ \\
\midrule
\textbf{Mistral} & & & & & & \\
\quad\textit{3B}    & .18/.05 & .26/.04 & .21/.05 & .54/.06 & .62/.05 & .56/.05 \\
\quad\textit{8B}    & .27/.01 & .29/.02 & .20/.08 & .65/.02 & .68/.04 & .56/.08 \\
\quad\textit{14B}   & .20/.04 & .26/.04 & .24/.06 & .59/.07 & .63/.07 & .59/.05 \\
\quad\textit{Large} & .02/.02 & .04/.03 & .03/.04 & .44/.05 & .46/.06 & .53/.07 \\
\midrule
\textbf{Qwen 3} & & & & & & \\
\quad\textit{0.6B}  & .00/.00 & .00/.00 & .00/.00 & .00/.00 & .00/.00 & .00/.00 \\
\quad\textit{1.7B}  & .01/.03 & .02/.03 & .02/.03 & .59/.05 & .63/.05 & .58/.04 \\
\quad\textit{4B}    & .00/.01 & .00/.01 & .00/.01 & .35/.02 & .35/.02 & .31/.03 \\
\quad\textit{8B}    & .28/.02 & .28/.02 & .30/.02 & .71/.03 & .72/.03 & .71/.02 \\
\quad\textit{14B}   & .32/.05 & .32/.07 & .32/.05 & .68/.06 & .70/.07 & .71/.05 \\
\quad\textit{32B}   & .21/.02 & .20/.04 & .24/.03 & .56/.04 & .58/.05 & .60/.04 \\
\bottomrule
\end{tabular}
\end{table}

\section{Parameter Estimation Procedure}
\label{app:estimation}

We estimate the parameters $(\alpha_0, \beta_0, \gamma, \rho)$ of the 
generative model separately for each model using the null condition 
(sneakers vs.\ sneakers) data. The procedure is a hybrid of maximum 
likelihood estimation and method of moments on the logit scale. Let $ \sigma = 1 / \sqrt{\gamma} $.

\paragraph{Step 1: Persona base rates.} For each persona $i$, compute 
$\hat{p}_i = \bar{y}_{i\cdot\cdot}$, the mean response across all 
perturbations and replicates. 

\paragraph{Step 2: Beta prior ($\alpha_0, \beta_0$).} Fit 
$\text{Beta}(\alpha_0, \beta_0)$ to the distribution of $\{\hat{p}_i\}$ 
across personas via maximum likelihood, using Nelder--Mead optimization 
of the negative log-likelihood.

\paragraph{Step 3: Logit-scale residuals.} For each (persona, perturbation) 
pair, compute the perturbation-level response rate $\hat{p}_{ij} = 
\bar{y}_{ij\cdot}$ (mean over replicates) and the logit-scale residual:
\begin{equation*}
    r_{ij} = \text{logit}(\hat{p}_{ij}) - \text{logit}(\hat{p}_i).
\end{equation*}

\paragraph{Step 4: Total variance ($\gamma$).} Compute 
$\hat{\sigma}^2 = \text{Var}(r_{ij})$ across all valid (persona, 
perturbation) cells, giving $\hat{\gamma} = 1/\hat{\sigma}^2$.

\paragraph{Step 5: Shared variance and intraclass correlation ($\rho$).} 
Compute perturbation means $\bar{r}_{\cdot j} = \frac{1}{N}\sum_i r_{ij}$ 
and apply the finite-$N$ bias correction:
\begin{equation*}
    \hat{\sigma}^2_u = \frac{N \cdot \text{Var}(\bar{r}_{\cdot j}) - 
    \hat{\sigma}^2}{N - 1},
\end{equation*}
clamped to $[0, \hat{\sigma}^2]$. The intraclass correlation is then 
$\hat{\rho} = \hat{\sigma}^2_u / \hat{\sigma}^2$.

\end{document}